\documentclass[IEEEtran]{ieeeconf}  

\IEEEoverridecommandlockouts                              

\newtheorem{mydef}{\bf Definition}
\newtheorem{mythm}{\bf Theorem}
\newtheorem{myprob}{\bf Problem}
\newtheorem{mylem}{\bf Lemma}
\newtheorem{mycol}{\bf Corollary}
\newtheorem{mypro}{\bf Proposition}
\newtheorem{myexm}{\bf Example}
\newtheorem{remark}{Remark}
\usepackage{amsmath}
\usepackage{mathrsfs}
\usepackage{amssymb}
\usepackage{algorithmic}

\usepackage{graphicx}
\usepackage[linesnumbered,ruled,vlined]{algorithm2e}

\def \R{\mathbb{R}_{\max}}

\title{\LARGE \bf
On the Computation   of Backward Reachable Sets for \\
Max-Plus Linear Systems with Disturbances}
\author{Yuda Li   and Xiang Yin%
\thanks{This work was supported by the National Science and Technology Major Project (2025ZD1600700) and the National Natural Science Foundation of China (62573291,62533017,62173226).}
	\thanks{Yuda Li and Xiang Yin are with School of Automation and Intelligent Sensing, Shanghai Jiao Tong University, Shanghai 200240, China.
	{\tt\small \{yuda.li,   yinxiang\}@sjtu.edu.cn}.}
}
\usepackage{color}

\begin{document}

\maketitle
\thispagestyle{empty}
\pagestyle{empty}

\begin{abstract}
This paper investigates one-step backward reachability for uncertain max-plus linear systems with additive disturbances. Given a target set, the problem is to compute the set of states from which there exists an admissible control input such that, for all admissible disturbances, the successor state remains in the target set. This problem is closely related to safety analysis and is challenging due to the high computational complexity of existing approaches.
To address this issue, we develop a computational framework based on tropical polyhedra. We assume that the target set, the control set, and the disturbance set are all represented as tropical polyhedra, and study the structural properties of the associated backward operators. In particular, we show that these operators preserve the tropical-polyhedral structure, which enables the constructive computation of reachable sets within the same framework. The proposed approach provides an effective geometric and algebraic tool for reachability analysis of uncertain max-plus linear systems. Illustrative examples are included to demonstrate the proposed method. 
\end{abstract}


\section{Introduction}
Max-plus linear systems (MPLSs) is an important class of discrete-event systems formulated over the max-plus algebra  \cite{hardouin2018control,de2020analysis}. This algebraic setting is particularly well suited for describing synchronization and delay effects in timed event-driven systems. MPLSs are closely connected with timed event graphs \cite{amari2011max,he2021performance}, which form a subclass of Petri nets characterized. Because of their effectiveness in representing event timings under precedence relations, max-plus models have been widely used in applications such as scheduling \cite{he2016cycle}, manufacturing \cite{hardouin2011towards,imaev2008hierarchial,van2020model}, and transportation systems \cite{kistosil2018generalized,kersbergen2016towards,outafraout2020control}.

In practical applications, however, max-plus systems are often affected by uncertainty arising from disturbances, modeling errors, or parameter variations. This makes reachability analysis and controller synthesis significantly more challenging. Existing approaches for uncertain MPLSs can be broadly divided into two categories. The first models uncertainty probabilistically. For example, \cite{soudjani2016formal} studied formal verification for stochastic max-plus linear systems by treating the system matrix as a random variable, while \cite{xu2018model} investigated model predictive control for max-plus linear systems with uncertain parameters endowed with probability distributions. These approaches provide a probabilistic characterization of uncertainty and enable the use of stochastic verification and control techniques.

A second line of work models uncertainty through deterministic ambiguity sets rather than probability distributions. In this direction, \cite{candido2018conditional} studied reachability analysis for uncertain max-plus linear systems subject to bounded noise, disturbances, and modeling errors. There, the uncertainty set is described in the interval max-plus framework, and reachability is computed by transforming the system into a piecewise affine representation and then applying difference-bound matrices (DBMs). However, the associated computations may become expensive as the system dimension and uncertainty structure grow.

In this paper, we study the backward reachability problem for   max-plus linear systems with disturbances. Given a target set, the objective is to compute the set of  states from which there exists an admissible control input such that, for all admissible disturbances, the successor state remains in the target set. This problem is closely connected to safety analysis and controlled invariance, since backward reachable sets provide the basic ingredient for determining safe sets and synthesizing safety-preserving controllers. While such problems can in principle be addressed using DBM-based techniques, the resulting computations may suffer from high complexity. Moreover, compared with the disturbance-free case studied in our recent work \cite{li2026computation}, the presence of disturbances introduces a substantial new difficulty through the universal quantification over the uncertainty set.

Our approach is motivated by recent developments in tropical convexity, in particular the theory of tropical polyhedra  \cite{allamigeon2013computing,espindola2025set}. We model the control and disturbance sets, as well as the target set, within a tropical polyhedral framework, which provides a flexible and expressive representation of uncertainty in the max-plus setting. Based on this representation, we develop a computational procedure for backward reachability analysis of uncertain max-plus linear systems. The proposed framework accommodates a broader class of uncertainty structures than interval-based descriptions, while preserving a compact geometric representation amenable to computation. By exploiting structural properties of tropical polyhedra, we show that backward reachable sets can still be characterized and computed within the same framework. In this way, the paper offers a new geometric perspective on the analysis of uncertain max-plus systems and provides an effective tool for safety verification and control synthesis.

\section{Preliminary}
In this section, we introduce the basic notation and preliminary results used throughout the paper.
\subsection{The Max-Plus Algebra $\mathbb{R}_{\max}$}
The max plus algebra $\mathbb{R}_{\max}$ is the set of real numbers together with an infinite point: $\mathbb{R} \cup \{-\infty\}$, together with the binary operations
$\oplus:\R\times\R\to\R$ and
$\otimes:\R\times\R\to\R$
defined by
$a\oplus b=\max(a,b)$ and $a\otimes b=a+b$.
For any $a\in\mathcal{D}$, we define
$(-\infty)\oplus a=a\oplus(-\infty)=a$ and
$(-\infty)\otimes a=a\otimes(-\infty)=-\infty$.
For simplicity, we write $\varepsilon$ for $-\infty$ throughout the paper.
On $\mathbb{R}_{\max}$, we consider the metric
$d(x,y)=|\exp(x)-\exp(y)|$,
with the convention that $\exp(\varepsilon)=0$.
Open and closed sets are understood with respect to the topology induced by this metric.
In the remainder of the paper, we focus on the max-plus algebra $\mathbb{R}_{\max}$, which we simply denote by $\R$.

\subsection{Semimodules, Sub-Semimodules, and Convex Cones}
The Cartesian product of $n$ copies of $\R$, denoted by $\R^n$, is naturally equipped with the vector addition
$\oplus:\R^n\times\R^n\to\R^n$ defined componentwise by
$(a_1,\dots,a_n)\oplus(b_1,\dots,b_n)
=
(a_1\oplus b_1,\dots,a_n\oplus b_n)$,
and the scalar multiplication
$\cdot:\R\times\R^n\to\R^n$ defined by
$a\cdot(b_1,\dots,b_n)
=
(a\otimes b_1,\dots,a\otimes b_n)$.
The triple $(\R^n,\oplus,\cdot)$ is called a max-plus $\R$-semimodule, or simply a semimodule.
This space is also endowed with the natural partial order
$x\le y
\Leftrightarrow
x_i\le y_i,\ \forall\,1\le i\le n$.
The element $\varepsilon_n=(\varepsilon,\dots,\varepsilon)\in\R^n$ is the minimal element.

Any matrix $A\in\R^{n\times m}$ defines a linear map
$\phi_A:\R^m\to\R^n$ by
$\phi_A(x)_i=\bigoplus_{j=1}^m A_{ij}\otimes x_j$,
where $ x=(x_1,\dots,x_m)\in\R^m$.
In the rest of the paper, we simply write $\phi_A(x)$ as $A\otimes x$.
Similarly, for $a,b\in\R^n$, their scalar product is defined by
$
(a|b)=\bigoplus_{i=1}^n a_i\otimes b_i$. 
A subsemimodule is a subset $\mathcal{S}\subseteq\R^n$ such that
$\forall\,x,y\in\mathcal{S},\ \forall\,\mu,\lambda\in\mathcal{D}
\mu\cdot x\oplus\lambda\cdot y\in\mathcal{S}$. 
The topology on $\R^n$ is taken to be the product topology induced by the metric on $\R$.
A convex cone over $\mathcal{D}$ is a subset $\mathcal{C}\subseteq\R^n$ satisfying
$\forall x,y \in \mathcal{D},\ \forall \lambda_1,\lambda_2 \in R,\ [\lambda_1\oplus \lambda_2 = 0] \Rightarrow \lambda_1\cdot x\oplus\lambda_2\cdot y \in \mathcal{D}$.
For two subsets $A,B\subseteq\R^n$, we define their max-plus sum by
$
A\oplus B=\{a\oplus b\mid a\in A,\ b\in B\}$.

\subsection{Tropical Cones, Half-Spaces and Polyhedra}
A tropical half-space $\mathscr{H}$ is a subset of $\R^n$ of the form 
$\mathscr{H} = \{x \in \R^n \mid (a|x) \le (b|x)\}$, where $a, b \in \R^n$. 
A tropical cone is an intersection of finitely many half-spaces, 
or equivalently, a set of the form 
\begin{equation}
    \mathcal{C} = \{x \in \R^n \mid A \otimes x \le B \otimes x\}
\end{equation}
which we denote for simplicity as $\mathcal{C} = \langle A, B \rangle$.
Similarly, a tropical polyhedra is a set of the form 
\begin{equation}
   \mathcal{P} = \{x 
\in \mathbb{R}_{\max}\mid A \otimes x \oplus c \le B \otimes x \oplus d\},
\end{equation}
which we denote for simplicity as $\mathcal{P} = \langle (A, c), (B, d) \rangle$. 
We call such representation an outer representation, or $\mathcal{M}$-form of a tropical cone (or tropical polyhedra).
For simplicity, we denote $\langle A,B \rangle^s = \langle  A,B\rangle \cap \langle  B,A\rangle$ (resp $\langle (A,c),(B,d) \rangle^s = \langle  (A,c),(B,d)\rangle \cap \langle  (B,d),(A,c)\rangle$), which are also tropical polyhedra. It has been proved by \cite{butkovivc2010max} that every tropical cone $\mathcal{C}$ can be generated by a finite set such that
\begin{equation}\label{V form of cone}
        \exists v_1, \dots, v_N \in \mathcal{C},\ \mathcal{C} = \left\{\left.\bigoplus_{i=1}^N \lambda_i \cdot v_i \right| \lambda_1, \dots, \lambda_N \in \R\right\}.
\end{equation}
We call $\{v_1, \dots, v_N\}$ the inner representation of $\mathcal{C}$, or its $\mathcal{V}$-form, and denote $\mathcal{C}$ as $\mathrm{Span}(\{v_1, \dots, v_N\})$. 
Tropical cones and tropical half-spaces are both subsemimodules, but the reciprocate is not true, since a subsemimodule may not be finitely generated.

A \textbf{bounded tropical polyhedron} is a finitely generated tropical convex cone of the following general form 
\[\mathcal{B} = \left\{\bigoplus_{i=1}^N \lambda_i\cdot e_i\left| \bigoplus_{i=1}^N\lambda_i=0\right\}\right.,\]
where $N \in \mathbb{Z}_{>0}$ and $e_i \in \R^n$ are vectors.  For simplicity, we write $\mathcal{B} = \mathrm{Conv}(e_1,\dots ,e_N)$.
It was proved in \cite{gaubert2011minimal} that $\mathcal{P}$ can be decomposed as
$\mathcal{P}=\mathcal{R}\oplus\mathcal{B}$, 
where $\mathcal{R}$ is a tropical cone and $\mathcal{B}$ is a bounded tropical polyhedron.
The tropical cone $\mathcal{R}$ is called the recession cone, and it is unique for this decomposition.

For a vector $v\in\R^n$, we denote by $v_{[ k: m]}$ the subvector consisting of the coordinates with indices ranging from $k$ to $m$.
For a matrix $A\in \R^{k\times m}$, we denote by $A_{[k_1:k_2][m_1:m_2]}$ the submatrix $(A_{i,j})_{k_1\le i\le k_2, m_1\le j\le m_2}$. For simplicity, if $k_1 = 1,k_2 = k$ or $m_1 = 1, m_2 = m$, we denote respectively $A_{\bullet,[m_1:m_2]}$ and $A_{[k_1:k_2],\bullet}$ (or simplier $A_{[k_1:k_2]}$) . If $k_1=k_2$ or $m_1 = m_2$, we simply denote as $A_{k_1,[m_1:m_2]}$ or $A_{[k_1:k_2],m_1}$. 

Next we discuss some useful properties about tropical cones and tropical polyhedra, which will be used throughout this paper. First,  both tropical cones and bounded tropical polyhedras are tropical polyhedras, since they are respectively a special case with $\mathcal{B} = \{\varepsilon_n\}$ or with $\mathcal{R} = \{\varepsilon_n\}$. 
The following lemma characterize the relation between tropical polyhedra and 
\begin{mylem}
    Let $\mathcal{P} \subseteq \R^n$ be a tropical polyhedra, then there exists $\mathcal{C} \subseteq \R^{n+1}$ a tropical cone such that 
    \[\mathcal{P} = \{x \in \R^n\mid [0,x^T]^T \in \mathcal{C}\}.
    \]
    Reciprocally, all sets of this form is a tropical polyhedra.
\end{mylem}
\begin{proof}
For simplicity, we identify the tuple $(x_1,\dots,x_n)$ with the column vector $[x_1,\dots,x_n]^T$, and we write $[0,x^T]^T$ simply as $(0,x)$.
Note that $\mathcal{P}$ is precisely the intersection of the hyperplane $x_1=0$ with the tropical cone $\mathcal{C}$.
Thus, by a slight abuse of notation, we write $\mathcal{P} = \mathcal{C}\cap \R^n$. Given a tropical polyhedra with the following form 
\[\mathcal{P} = \mathrm{Span}(v_1,\dots,v_N) \oplus \mathrm{Conv}(e_1,\dots,e_M)),\]
we can construct the corresponding tropical cone 
\[\mathcal{C} = \mathrm{Span}((\varepsilon,v_1),\dots,(\varepsilon,v_N),(0,e_1),\dots,(0,e_M)) \subseteq \R^{n+1}\]
then $\mathcal{C}$ is such that $\mathcal{P} = \mathcal{C}\cap \R^n$. Conversely, given a tropical cone $\mathcal{C}\subseteq\R^{n+1}$ generated by a set $V$, define
\[
U
=
\{(-v_1)\cdot v_{[2:n+1]}\in\R^n \mid v\in V,\ v_1\neq\varepsilon\},
\]
and
\[
W
=
\{v_{[2:n+1]}\in\R^n \mid v\in V,\ v_1=\varepsilon\}.
\]
Then one can verify that
\[
\mathcal{C}\cap\R^n
=
\mathrm{Span}(W)\oplus\mathrm{Conv}(U).
\]    
\end{proof}

\section{Problem Formulation}\label{sec: problem formulation}

In this section, we formulate the backward reachability problem for max-plus linear systems subject to disturbances. We first introduce the system model and the associated backward reachability operator, and then decompose this operator into several elementary set-valued operators that will be studied in the sequel.

We consider the following max-plus linear  system (MPLS) with disturbance over $\R$:
\begin{equation}\label{system disturbance}
    x_k = A \otimes x_{k-1} \oplus B \otimes u_k \oplus C \otimes w_k,
\end{equation}
where $x_k \in \R^n$ is the system state, $u_k \in \mathcal{U} \subseteq \R^m$ is the control input, and $w_k \in \mathcal{W} \subseteq \R^q$ is the disturbance. The matrices $A \in \R^{n\times n}$, $B \in \R^{n\times m}$, and $C \in \R^{n\times q}$ are the system matrices.

Given a target subset $E \subseteq \R^n$, we define the (one-step) backward reachable set of $E$ as
\[
    \Upsilon(E)
    =
    \left\{
    x \in \R^n \ \middle|\
    \begin{aligned}
        &\exists u \in \mathcal{U},\ \forall w \in \mathcal{W}\text{ s.t. }\\
        &A \otimes x \oplus B \otimes u \oplus C \otimes w \in E
    \end{aligned}
    \right\}.
\]
In other words, $\Upsilon(E)$ consists of all states from which one can choose a control input such that, for every admissible disturbance, the successor state remains in $E$ after one step.

To better understand the structure of the backward reachability operator $\Upsilon$, we next introduce three elementary set-valued operators that capture, respectively, the effect of the system dynamics, the universal quantification over disturbances, and the existential quantification over control inputs.
\begin{itemize}
    \item 
    \textbf{Inverse operator}: 
    $A^{-1}: 2^{\R^n} \to 2^{\R^n}$ such that 
    \[
    E \mapsto \{x \in \R^n \mid A \otimes x \in E\}.
    \]
    \item 
    \textbf{Universal inverse operator} (w.r.t.\ the disturbance set $\mathcal{W}$):  
    $ \phi_{\mathcal{W}}: 2^{\R^n} \to 2^{\R^n}$ such that 
    \[
    E \mapsto \{x \in \R^n \mid \forall w \in \mathcal{W},\ x \oplus C \otimes w \in E\}.
    \]
    \item 
    \textbf{Existential inverse operator}  (w.r.t\ the control set $\mathcal{U}$):  
    $\gamma_{\mathcal{U}}: 2^{\R^n} \to 2^{\R^n}$ such that 
    \[ 
    E \mapsto \{x \in \R^n \mid \exists u \in \mathcal{U},\ x \oplus B \otimes u \in E\}.
    \] 
\end{itemize}

With these definitions, the backward reachability operator can be decomposed as
\[
    \Upsilon = A^{-1} \circ \gamma_{\mathcal{U}} \circ \phi_{\mathcal{W}}.
\]
This decomposition separates the contributions of the disturbance, the control action, and the autonomous system dynamics, and will serve as the basis for the computational developments in the subsequent sections.

We now state the main problem considered in this paper.

\begin{myprob}[\bf Backward Reachability
Problem]
Let $\mathscr{S}$ be the MPLS   with disturbance defined by \eqref{system disturbance}. Assume that the control set $\mathcal{U}\subseteq \R^m$, the disturbance set $\mathcal{W}\subseteq \R^q$, and the target set $S\subseteq \R^n$ are tropical polyhedra.
Given a tropical polyhedron $E\subseteq \R^n$, compute the (one-step) backward reachable set $\Upsilon(E)$ of $E$ with respect to $\mathscr{S}$.
\end{myprob}

Compared with classical linear systems, the backward reachability problem considered here is substantially more involved. In the max-plus framework, the operations defining the system evolution may transform a set in a way that changes its underlying geometric structure. In particular, even when the initial set is a tropical polyhedron, it is not obvious a priori whether the successive application of inverse images, universal constraints with respect to disturbances, and existential projections with respect to controls still yields a tropical polyhedron. This motivates the main objective of this paper: \textbf{to characterize conditions under which the operator $\Upsilon$ preserves the tropical polyhedral structure, and to develop effective procedures for computing $\Upsilon(E)$ when $E$, $\mathcal{U}$, and $\mathcal{W}$ are tropical polyhedra.} 

The remainder of this paper addresses the solution to Problem~1 and is organized as follows. Section~\ref{sec: problem solution} derives computational procedures for the operators $A^{-1}$ and $\gamma_{\square}$, which are based on similar techniques. In contrast, the universal inverse operator $\phi_{\square}$ is more challenging and is therefore treated separately. Section~\ref{sec: computation for phi} analyzes its structural properties, and Section~\ref{sec: computation phi} develops a computational method based on these results.

\section{Computation for the operators $A^{-1}$ and $\gamma_\square$}\label{sec: problem solution}
In the following, we explain how to compute the backward reachable set associated with the inverse operator $A^{-1}$ and the existential inverse operator $\gamma_\square$ defined in Section \ref{sec: problem formulation}. 
We prove that the tropical polyhedra structure 
is preserved under both operators and derive a method to determine generating vectors for the image of a tropical polyhedron under each of these operators. 

We begin with the following lemma, which is essential for the computation of these operators.
\begin{mylem}\label{projection of convex cone}
    Let $Z$ be a tropical cone, then the projection on its first $r$ coordinates ($r\le n$):
    \[p_r(Z) = \{z\in \R^r\mid \exists v \in \R^{n-r},\ (z,v) \in Z \}\]
     is still a tropical cone.
\end{mylem}
\begin{proof}
Let $v_1,\dots,v_k$ be a family of generating vectors for $\mathcal{C}$. Denote by $v_1',\dots,v_k'$ the vectors obtained by restricting $v_1,\dots,v_k$ to their first $r$ coordinates. Then the projection of $Z$ onto the first $r$ coordinates is the tropical cone in $\R^{r}$ generated by $v_1',\dots,v_k'$.
\end{proof}
\subsection{Computation for the image of $A^{-1}$}\label{subsec: A-1}
In this section, we illustrate the procedure for computing the backward propagation set $A^{-1}(Z)$, where $A \in \R^{n \times n}$ is a matrix and $Z \subseteq \R^n$ is a tropical polyhedron of the form $\mathcal{C} \cap \R^n$. Here, $\mathcal{C} \subseteq \R^{n+1}$ is a tropical cone generated by a matrix $M \in \R^{(n+1)\times q}$.

Without loss of generality, we assume that all coefficients in the first row of $M$ are either $\varepsilon$ or $0$. By definition, we have
\[
A^{-1}(Z) = \left\{x \in \R^n\left|
\begin{aligned}
    &\exists u \in \R^q, A\otimes x = M_{[2:n+1],\bullet}\otimes u \\
    &M_{1,\bullet}\otimes u = 0
\end{aligned}
\right\}\right.
\]

Note that $A^{-1}(Z)$ can be equivalently rewritten as
\[
A^{-1}(Z) = \left\{x \in \R^n\left|
\begin{aligned}
    &\exists u \in \R^q,\ z = (0,x,u)\\
    &\mathcal{M}_1^{A,M}\otimes z = \mathcal{M}_1^{A,M}\otimes z
\end{aligned}
\right\}\right.
\]
where the matrices $\mathcal{M}_1^{A,M}$ and $\mathcal{M}_2^{A,M}$ are defined by
\[
\mathcal{M}_1^{A,M} =
\left[
\begin{array}{c|c|c}
    0 & \mathcal{E}^{1\times n} & \mathcal{E}_{1\times q}\\ \hline
    \mathcal{E}_{n\times 1} & A & \mathcal{E}_{n\times q}
\end{array}
\right]
\]
\[
\mathcal{M}_2^{A,M} =
\left[
\begin{array}{c|c|c}
    \varepsilon & \mathcal{E}^{1\times n} & \mathcal{M}_{1,\bullet}\\ \hline
    \mathcal{E}_{n\times 1} & \mathcal{E}_{n\times n} & \mathcal{M}_{[2:n+1],\bullet}
\end{array}
\right]
\]

Observe that the vector $(0,x,u)$ ranges precisely over the tropical polyhedron $\left\langle \mathcal{M}_1^{A,M}, \mathcal{M}_2^{A,M} \right\rangle$. Consequently, we obtain
\[
A^{-1}(Z) = \left\{x \in \R^n\left| (0,x) \in p_{n+1}\left(\left\langle \mathcal{M}_1^{A,M}, \mathcal{M}_2^{A,M} \right\rangle^s\right)\right\}\right.
\]
\[
= p_{n+1}\left(\left\langle \mathcal{M}_1^{A,M}, \mathcal{M}_2^{A,M} \right\rangle^s\right) \cap \R^n
\]

This formulation provides an outer description of the set $A^{-1}(Z)$, and the above construction leads to the following corollary.

\begin{mycol}
Let $A \in \R^{n \times n}$ be a matrix, and let $Z \subseteq \R^n$ be a tropical polyhedron. Then the set $A^{-1}(Z)$ is also a tropical polyhedron.
\end{mycol}

\subsection{Computation for $\gamma_\square$}
In this section, we illustrate the computation of the backward reachable set associated with the operator $\gamma_\square$. Let $A \in \R^{n\times n}$ be a matrix, and let $\mathcal{U} = \mathcal{G}\cap \R^n$ be a tropical polyhedron, where $\mathcal{G}$ is generated by a matrix $R \in \R^{(n+1)\times r}$ whose first-row coefficients are either $0$ or $\varepsilon$. Let $Z$ be a tropical polyhedron of the same form as defined in Section~\ref{subsec: A-1}. Our goal is to compute the set $\gamma_\mathcal{U}(Z)$. By definition, we have
\begin{equation}\nonumber
    \gamma_\mathcal{U}(Z) = \left\{x \in \R^m\left| 
\begin{aligned}
    & \exists y \in \R^r,\ \exists z \in \R^q, \\
    & x \oplus A \otimes R_{[2:n+1],\bullet} \otimes y = \\
    &\qquad \qquad M_{[2:n+1],\bullet} \otimes z,\\
    & R_{1,\bullet}\otimes y = 0,\quad M_{1,\bullet} \otimes z = 0
\end{aligned}
\right.\right\}.
\end{equation}

Similarly to Section~\ref{subsec: A-1}, the set $\gamma_\mathcal{U}(Z)$ can be rewritten as
\[
\gamma_\mathcal{U}(Z) = \left\{x \in \R^m \left|
\begin{aligned}
    &\exists (y,z) \in \R^r \times \R^q\\
    &\exists t=(0,x,y,z), \\
    &\mathcal{M}_1^{A,R,M}\otimes t = \mathcal{M}_2^{A,R,M} \otimes t
\end{aligned}
\right.\right\}.
\]

The matrices $\mathcal{M}_1^{A,R,M}$ and $\mathcal{M}_2^{A,R,M}$ are defined as
\[
\mathcal{M}_1^{A,R,M} = \left[
\begin{array}{c|c|c|c}
    \varepsilon & \mathrm{Id}_m & A\otimes R_{[2:n+1],\bullet} & \mathcal{E}_{m\times q} \\ \hline
    \varepsilon & \mathcal{E}_{1\times m} & R_{1,\bullet} & \mathcal{E}_{1\times q} \\ \hline
    \varepsilon & \mathcal{E}_{1\times m} & \mathcal{E}_{1\times r} & M_{1,\bullet}
\end{array}
\right]
\]
\[
\mathcal{M}_2^{A,R,M} = \left[
\begin{array}{c|c|c|c}
    \varepsilon & \mathcal{E}_{m\times m} & \mathcal{E}_{m\times r} & M_{[2:n+1],\bullet} \\ \hline
    0 & \mathcal{E}_{1\times m} & \mathcal{E}_{1\times r} & \mathcal{E}_{1\times q} \\ \hline
    0 & \mathcal{E}_{1\times m} & \mathcal{E}_{1\times r} & \mathcal{E}_{1\times q}
\end{array}
\right].
\]

Then $\gamma_\mathcal{U}(Z)$ can be expressed as
\[
\gamma_\mathcal{U}(Z) = p_{n+1}\left(\left\langle \mathcal{M}_1^{A,R,M}, \mathcal{M}_2^{A,R,M} \right\rangle^s\right) \cap \R^n.
\]

As before, this formulation provides an outer description of the set $\gamma_\mathcal{U}(Z)$. We thus obtain the following corollary.

\begin{mycol}
Let $A \in \R^{n \times n}$ be a matrix, and let $Z,\mathcal{U}$ be two tropical polyhedra. Then the set $\gamma_\mathcal{U}(Z)$ is also a tropical polyhedron.
\end{mycol}

\section{Computation for $\phi_\mathcal{W}$}\label{sec: computation for phi}

We are now interested in the set propagation operator $\phi_\square$. Given a tropical polyhedron $Z$ of the form $\langle C,D \rangle \cap \R^n$ (with $C,D \in \R^{q\times (n+1)}$), we aim to characterize the set $\phi_{\mathcal{W}}(Z)$, where $\mathcal{W}$ is a tropical polyhedron in $\R^n$. In the following, we first show that this set is itself a tropical polyhedron, and then provide techniques to derive its inner representation.

\subsection{Simplification Step}

As shown in \cite{gaubert2011minimal}, $\mathcal{W}$ can be decomposed as the max-plus sum of two components: $\mathcal{W} = \mathcal{R} \oplus \mathcal{P}$, where $\mathcal{P}$ is a bounded tropical polyhedron and $\mathcal{R}$ is a tropical cone (its recession cone). The following proposition allows us to simplify the computation of $\phi_\mathcal{W}(Z)$.
\begin{mypro}
Let $Z$ and $\mathcal{W} = \mathcal{P}\oplus \mathcal{R}$ be defined as above. A necessary condition for $\phi_\mathcal{W}(Z)$ to be nonempty is that $\mathcal{R}$ is contained in the tropical cone $\langle C_{\bullet,[2:n+1]},D_{\bullet,[2:n+1]}\rangle$. In this case, we have $\phi_\mathcal{W}(Z) = \phi_{\mathcal{P}}(Z)$.
\end{mypro}
\begin{proof}
    We first prove that $\mathcal{R}\subseteq \langle C_{\bullet,[2:n+1]},D_{\bullet,[2:n+1]}\rangle$ is a necessary condition for $\phi_\mathcal{W}(Z)$ being non empty. Suppose by contradiction that there exists $y \in \mathcal{R}$ such that $C_{\bullet,[2:n+1]}\otimes y \le D_{\bullet,[2:n+1]}\otimes y$ does not holds 
    (there exists index $i$, $(C_{i,[ 2:n+1]}|y)>(D_{i,[2:n+1]}|y)$) and there exists $x \in \R^n$ s.t. $\forall u \in \mathcal{W}$, $x \oplus u \in Z$, then for arbitrary $y' \in \mathcal{P}$ and for all $\lambda \in \R$, we have $C\otimes (0,x \oplus y' \oplus \lambda \cdot y) \le D \otimes (0,x \oplus y' \oplus \lambda \cdot y)$. As $(C_{i,[2,n+1]} |\lambda \cdot  y) > (D_{i,[2:n+1]}|\lambda \cdot y)$ for all $\lambda > \varepsilon$, fix $y' \in \mathcal{P}$, then for $\lambda$ large enough we will have $(C_i|(0,x\oplus y' \oplus \lambda \cdot y)) > (D_i |(0,x\oplus y' \oplus \lambda \cdot y))$, which is a contradiction. Once we have verified that the recession cone $\mathcal{R}$ is contained in $(C_{i,[ 2:n+1]}|y)>(D_{i,[2:n+1]}|y)$, we can prove that $\phi_{\mathcal{W}}(Z) = \phi_{\mathcal{P}}(Z)$. Indeed, it is immediate that $\phi_\mathcal{W}(Z) \subseteq \phi_{\mathcal{P}(Z)}$. To show that $\phi_\mathcal{P}(Z) \supseteq \phi_{\mathcal{P}}(Z)$, take $x \in \phi_{\mathcal{P}}(Z)$, then for all $y' \in \mathcal{P}$, we have $C \otimes(0,x \oplus y')\le D\otimes (0,x \oplus y') $. Let $u \in \mathcal{W}$ be of the form $u = y \oplus y'$ with $y' \in \mathcal{P}$ and $y\in \mathcal{R}$, since $C_{\bullet,[2:n+1]}\otimes y \le D_{\bullet,[2:n+1]}\otimes y$, we have $C\otimes (0,x \oplus y'\oplus y) \le D \otimes (0,x \oplus y' \oplus y)$, hence $C\otimes (0,x \oplus u) \le D \otimes (0,x \oplus u)$ which proves that $x \in \phi_\mathcal{W}(Z)$.
\end{proof}

\begin{myexm}
    Let us consider $\mathcal{W'}\subseteq \R^3$ the tropical cone generated by the following set $$\{[0,1,1]^T,[0,3,1]^T,[0,1,3]^T,[\varepsilon,0,0],[\varepsilon,1,0]\}$$
    Then $\mathcal{W} = \mathcal{W}'\cap \R^2$ is the tropical polyhedra of the following form:
    \[\mathcal{W} = 
    \mathrm{Span}\left(
    \begin{bmatrix}
        0\\ 0
    \end{bmatrix}
    ,
    \begin{bmatrix}
        1\\ 0
    \end{bmatrix}
    \right)
    \oplus
    \mathrm{Conv}\left(
    \begin{bmatrix}
        1\\1
    \end{bmatrix}
    ,
    \begin{bmatrix}
        3\\ 1
    \end{bmatrix}
    ,
    \begin{bmatrix}
        1\\3
    \end{bmatrix}
    \right)\]
    where the first part is the recession cone $\mathcal{R}$ of $\mathcal{W}$ and the second part is the bounded polyhedra $\mathcal{P}$. We consider the safety set $S$ in $\R^2$ defined by $S = \mathrm{Span}([1,0]^T,[0,1]^T)$, or equivalently, in its $\mathcal{M}$-form
\[S = \left\langle 
\begin{bmatrix}
    \varepsilon & \varepsilon & 0 \\
    \varepsilon & -1 &\varepsilon
\end{bmatrix} , 
\begin{bmatrix}
    \varepsilon& 1 & \varepsilon \\
    \varepsilon & \varepsilon & 0
\end{bmatrix}
\right\rangle
\cap \R^2
\]
which we denote as $S = \langle C,D\rangle\cap \R^2$. We can verifty that $(C_{1,\ge 2}|[0,0]^T) = 0 \le (D_{1,\ge 2}|[0,0]^T) = 1$, and $(C_{2,\ge2}|[1,0]^T) = -1 \le (D_{2,\ge2}|[1,0]^T) = 0$, thus $\mathcal{R} \subseteq S$, which implies that $\phi_\mathcal{U}(S)$ may not be empty and we have $\phi_\mathcal{U}(S) = \phi_\mathcal{R}(S)$.
\end{myexm}

\subsection{Reform $\phi_\mathcal{P}(Z)$ as An Intersection of Pseudo Half-Spaces}
Now the problem reduces to characterizing $\phi_\mathcal{P}(Z)$, where $\mathcal{P}$ is a bounded polyhedron. Let $\mathcal{P} = \mathrm{Conv}(e_1,\dots,e_M)$, and define $\mathcal{P}' = \mathrm{Span}((0,e_i)i)$. Then we have $\mathcal{P} = \mathcal{P}'\cap \R^n$. The set $\phi_\mathcal{P}(Z)$ can therefore be rewritten as
\[\phi_\mathcal{P}(Z) = \left\{x \in \R^{n+1}\left|
\begin{aligned}
    &\forall y \in \mathcal{P}',\ \text{s.t. } y_1 = x_1\\
    &C\otimes (x\oplus y) \le D \otimes (x \oplus y)
\end{aligned}
\right\}\right.\cap \R^n\]

To better understand its structure, we introduce the following definition.
\begin{mydef}[\bf Pseudo Half-Space]
A pseudo half-space  is a set of the form 
\[\mathscr{H}^\mathcal{U}_{c,d} = \left\{x \in \R^{n+1}\left|
\begin{aligned}
    &\forall y \in \mathcal{U},\ y_1 = x_1 \\
    &(c|x\oplus y)\le (d|x\oplus y),
\end{aligned}
\right\}\right.\]
where $c,d \in \R^{n+1}$ are two vectors and $\mathcal{U}$ is a tropical cone.
\end{mydef}

With this definition, the set $\phi_\mathcal{P}(Z)$ has the form $\phi_\mathcal{P}(Z) = \left(\bigcap_{i=1}^q\mathscr{H}_{C_{i},D_{i}}^{\mathcal{P}'}\right)\cap\R^n$. We now study some structural properties of this set.
\begin{mylem}
    The set $\mathscr{H}_{C_{i},D_{i}}^{\mathcal{P}'}$ is a closed sub-semimodule. In particular $E = \bigcap_{i=1}^q\mathscr{H}_{C_{i},D_{i}}^{\mathcal{P}'}$ is a closed sub-semimodule.
\end{mylem}
\begin{proof}
We first prove that $\mathscr{H}{C{i},D_{i}}^{\mathcal{P}'}$ is a sub-semimodule. Indeed, for all $\lambda \in \R\setminus \{\varepsilon\}$ and $x \in \mathscr{H}_{C_{i},D_{i}}^{\mathcal{P}'}$, for all $y \in \mathcal{P}'$ such that $(\lambda\cdot x)_1 = y_1$, we have $x_1 = ((-\lambda)\cdot y)_1$ and $(-\lambda)\cdot y \in \mathcal{P}'$, thus $(C_i | (x \oplus ((-\lambda)\cdot y))) \le (D_i |(x \oplus ((-\lambda)\cdot y)))$, which is equivalent to $(C_i | (\lambda\cdot x \oplus y))\le (D_i |(\lambda\cdot x \oplus y))$. For $\lambda = \varepsilon$, we have $\lambda\cdot x = \varepsilon_{n+1}$, thus $y_1=\varepsilon$, which implies $y = \varepsilon_{n+1}$, we then have $\varepsilon = (C_i| (\lambda\cdot x \oplus y)) \le (D_i|  (\lambda\cdot x \oplus y))= \varepsilon$.  Now let $x,x' \in  \mathscr{H}_{C_{i},D_{i}}^{\mathcal{P}'}$, if $x_1 = x_1'$, then it is trivial that $x\oplus x' \in  \mathscr{H}_{C_{i},D_{i}}^{\mathcal{P}'}$. Suppose now $x_1<x_1'$, let $y\in \mathcal{P}'$ such that $y_1 = (x\oplus x')_1$, then $y_1 = x'_1>\varepsilon$. We let $y' = (x_1-x'_1)\cdot y$. We verifies quickly that $y' \le y$ and $x \oplus y' \in \langle C_i,D_i \rangle$. Since $x' \oplus y \in \langle C_i,D_i \rangle$, we have $(x \oplus y') \oplus (x' \oplus y) = (x \oplus x') \oplus y \in \langle C_i,D_i \rangle$. We have proved that $ \mathscr{H}_{C_{i},D_{i}}^{\mathcal{P}'}$ is a subsemimodule. Further more, we can show that $ \mathscr{H}_{C_{i},D_{i}}^{\mathcal{P}'}$ is closed. Indeed, let $\{x^{(n)}\}_{n\in \mathbb{Z}_{>0}}$ be a sequence in $ \mathscr{H}_{C_{i},D_{i}}^{\mathcal{P}'}$ converging to a point $x \in \R^{n+1}$ and let $y \in \mathcal{U}$ such that $y_1 = x_1$. We distinguish two cases: 
\begin{enumerate}
    \item If $x_1\neq \varepsilon$ then $y^{(n)} = (x_1-x_1^{(n)})\cdot y$ is a sequence in $\mathcal{P}'$ converging to $y$ and is such that $y^{(n)}_1=x^{(n)}_1$ for all $n$. Since $x^{(n)} \oplus y^{(n)} \in \langle C,D\rangle$ and $\langle C_i,D_i\rangle$ is a closed set, we conclude that $x \oplus y \in \langle C_i,D_i\rangle$.
    \item If $x_1 = \varepsilon$ then $y = \varepsilon_{n+1}$. Take an arbitrary $y^{(0)} \in \mathcal{U}$ s.t. $y_1^{(0)} = 0$. For $n \in \mathbb{Z}_{>0}$, define $y^{(n)} = x_1^{(n)}\cdot y^{(0)}$, then $\{y^{(n)}\}_n$ converge to $y$ and we have $\{x^{(n)} \oplus y^{(n)}\}$ converge to $x\oplus y = x$, which proves that $x \in \mathscr{H}_{C_{i},D_{i}}^{\mathcal{P}'}$.
\end{enumerate}

\end{proof}

Another noteworthy property is that the complement of $\mathscr{H}^{\mathcal{U}}_{c,d}$ is also stable under addition and scalar multiplication. Indeed, we have
\[
\left(\mathscr{H}^{\mathcal{U}}_{c,d}\right)^c = \left\{x \in \R^{n+1}\left|
\begin{aligned}
    &\exists y \in \mathcal{U},\ y_1 = x_1 \\
    &(c|x\oplus y)> (d|x\oplus y)
\end{aligned}
\right\}\right.
\]

Let $x,x' \in \left(\mathscr{H}^{\mathcal{U}}_{c,d}\right)^c$. Then there exist $y,y' \in \mathcal{U}$ such that $y_1 = x_1, y'_1 = x'_1$ and $(c|x\oplus y) > (d|x \oplus y), (c|x'\oplus y') > (d|x' \oplus y')$ . We then have
\begin{align}
&(c|(x\oplus x')\oplus (y\oplus y')) 
 = (c|x\oplus y) \oplus (c|x' \oplus y') \nonumber\\
 > &(d|x\oplus y) \oplus (d|x' \oplus y')  
 = (d|(x\oplus x')\oplus (y\oplus y')),  \nonumber 
\end{align} 
and moreover $(y_1\oplus y_2)^{(1)} = (x_1\oplus x_2)^{(1)}$. This shows that $x_1\oplus x_2 \in \left(\mathscr{H}^{\mathcal{U}}_{c,d}\right)^c$.

Now let $\lambda \in \R$ be a scalar. We readily verify that $(\lambda\cdot y)_1 = (\lambda \cdot x)_1$ and $(c|\lambda\cdot x \oplus \lambda \cdot y) > (d|\lambda\cdot x \oplus \lambda \cdot y)$. Therefore, $\lambda \cdot x \in \left(\mathscr{H}^{\mathcal{U}}_{c,d}\right)^c$.

\subsection{Induction Step for Calculating Generating Set for $\phi_\mathcal{P}(Z)$}

In the previous section, we established that the set $E = \bigcap_{i=1}^q \mathscr{H}^\mathcal{U}_{C_i,D_i}$ possesses a favorable mathematical structure. Building on these results, the goal of this section is to prove that $E$ is finitely generated, i.e., that it is a tropical cone, and to develop an algorithm for computing a generating set. The following theorem plays a key role in the recursive procedure underlying our algorithm.

\begin{mythm}\label{thm: Generating set for C cap H}
    Let $\mathscr{C} \subseteq \R^{n+1}$ be a closed tropical cone generated by a set $V$ of elements of $\R^{n+1}$, and let $\mathscr{H}^\mathcal{U}_{c,d}$ be a pseudo tropical half-space. Then the cone $\mathscr{C}\cap \mathscr{H}^{\mathcal{U}}_{c,d}$ is generated by the following set:
    \[(V \cap \mathscr{H}_{c,d}^\mathcal{U}) \cup \left\{v \oplus \rho\cdot w \left|
    \begin{aligned}
        &(v,w) \in (V\cap \mathscr{H}_{c,d}^{\mathcal{U}})\times (V \setminus\mathscr{H}^\mathcal{U}_{c,d})\\
        &\rho =  \max\{\lambda\mid v \oplus \lambda\cdot w \in \mathscr{H}_{c,d}^\mathcal{U}\}
    \end{aligned}
    \right\}\right.\]
\end{mythm}
 
\begin{proof}
Note that $\rho$ is well defined, in the sense that the set $M = \{\lambda \mid v \oplus \lambda\cdot w \in \mathscr{H}^{\mathcal{U}}_{c,d}\}$ admits a maximal element. Indeed, $M$ is nonempty since $\varepsilon \in M$. We claim that $M$ also admits a greatest element. Suppose, by contradiction, that this is not the case. Then $(-\lambda)\cdot v \oplus w \in \mathscr{H}^\mathcal{U}_{c,d}$ holds for arbitrarily large $\lambda$. Letting $\lambda \to +\infty$, we obtain that $(-\lambda)\cdot v \oplus w \in \mathscr{H}^{\mathcal{U}}_{c,d}$ converges to $w \notin \mathscr{H}^{\mathcal{U}}_{c,d}$, which contradicts the fact that $\mathscr{H}^{\mathcal{U}}_{c,d}$ is closed.
    
Now, we observe that the set defined in Theorem \ref{thm: Generating set for C cap H} is contained in $\mathscr{C}\cap \mathscr{H}^{\mathcal{U}}_{c,d}$. Let $x \in \mathscr{C}\cap \mathscr{H}^{\mathcal{U}}_{c,d}$. By the tropical analogue of the Minkowski--Carathéodory theorem \cite{gaubert2007minkowski}, $x$ can be expressed as a combination of at most $n+1$ elements of $V$. That is, there exist $V' \subseteq V \cap \mathscr{H}^{\mathcal{U}}_{c,d}$ and $W' \subseteq V \setminus \mathscr{H}^{\mathcal{U}}_{c,d}$ such that $|V'|+ |W'| \le n+1$, and
\[
x = \bigoplus_{v \in V' }\lambda_v \cdot v \oplus \bigoplus_{w\in W'}\lambda_w \cdot w
\]
where $\lambda_v,\lambda_w \in \R\setminus \{\varepsilon\}$. Denote $\rho_{(v,w)} = \max\{\lambda\mid v \oplus \lambda\cdot w\in \mathscr{H}^{\mathcal{U}}_{c,d}\}$.

We  show that for every $w \in W'$, there exists $v \in V'$ such that $\lambda_v\otimes \rho_{(v,w)}\ge \lambda_w$. Suppose, by contradiction, that there exists $w \in W'$ such that $\lambda_v \otimes \rho_{(v,w)} < \lambda_w$ for all $v \in V'$. Then $\lambda_v \cdot v \oplus \lambda_w\cdot w \notin \mathscr{H}^{\mathcal{U}}_{c,d}$ (since $\lambda_v \neq \varepsilon$). Since the complement of $\mathscr{H}^{\mathcal{U}}_{c,d}$ is stable under addition and multiplication by a nonzero scalar, it follows that
\[
x = \left(\bigoplus_{v \in V'}\lambda_v\cdot v\oplus \lambda_w\cdot w\right) \oplus \left(\bigoplus_{w' \in W' \setminus\{w\}}\lambda_{w'}\cdot w'\right)
\]
does not belong to $\mathscr{H}^{\mathcal{U}}_{c,d}$, which contradicts the assumption that $x \in \mathscr{H}^{\mathcal{U}}_{c,d}$.

For each $w \in W'$, let $v_w$ be an element of $V'$ such that $\lambda_{v_w}\otimes \rho_{(v_w,w)} \ge \lambda_w$. Since $\lambda_w \neq \varepsilon$, we have $\rho_{(v_w,w)}\neq \varepsilon$. Hence,
\[
\begin{aligned}
    x =& \left(\bigoplus_{v\in V'} \lambda_v \cdot v\right)\oplus \\ &\left(\bigoplus_{w\in W'}(\lambda_w-\rho_{(v_w,w)})\cdot (v_w\oplus \rho_{(v_w,w)}\cdot w)\right).
\end{aligned}
\]
\end{proof}
To compute a generating set for $E$, we proceed as follows. We start with $V = \{e_i\}_{i=1,\dots,n+1}$, which is a generating set for $\R^{n+1}$. We then construct the generating set inductively: if $V_j$ is a generating set for $\mathscr{C}_j = \bigcap_{i=1}^j\mathscr{H}^{\mathcal{U}}_{C_i,D_i}$, then Theorem \ref{thm: Generating set for C cap H} allows us to derive a generating set for
\(
\bigcap_{i=1}^{j+1}\mathscr{H}^{\mathcal{U}}_{C_i,D_i} = \mathscr{C}_j \cap \mathscr{H}^{\mathcal{U}}_{C_{j+1},D_{j+1}}.
\)
However, a key difficulty is that the set described in Theorem \ref{thm: Generating set for C cap H} is not straightforward to compute. The next section addresses this issue.

\section{Computation methods for $\phi_\square$}\label{sec: computation phi}

To compute a generating set for $E$, we need to determine the set $V \cap \mathscr{H}^{\mathcal{U}}_{c,d}$. Moreover, for $v \in V\cap \mathscr{H}^{\mathcal{U}}_{c,d}$ and $w \in V \setminus \mathscr{H}^{\mathcal{U}}_{c,d}$, we need to compute
\[
\rho = \max\{\lambda \mid v \oplus \lambda\cdot w \in \mathscr{H}^{\mathcal{U}}_{c,d}\}.
\]

We now explain how to determine, for a given $v \in V$, whether $v \in \mathscr{H}^{\mathcal{U}}_{c,d}$. The case $v_1 = \varepsilon$ is trivial. Suppose now that $v_1 \neq \varepsilon$, and for simplicity assume that $v_1 = 0$. Then $v \in \mathscr{H}^{\mathcal{U}}_{c,d}$ if and only if for all $y \in \mathcal{U}$ such that $y_1 = 0$, we have $(c|v\oplus y)\le (d|v \oplus y)$.

Let $U$ be a generating set for $\mathcal{U}$, and assume for simplicity that for all $u \in U$, $u_1 = 0$. A necessary condition for $v \in V \cap \mathscr{H}^{\mathcal{U}}_{c,d}$ is that for all $u \in U$, $(c|v \oplus u) \le (d|v \oplus u)$. In fact, this condition is also sufficient. Suppose it holds, and let $y = \bigoplus_{u \in U}\lambda_u\cdot u$, with $\bigoplus_{u} \lambda_u=0$. We show that $(c|v \oplus y) \le (d|v \oplus y)$ by considering two cases:

\begin{enumerate}
    \item If $(c|v) > (d|v)$, then for all $u \in U$, since $(c|v \oplus u) \le (d| v \oplus u)$, we must have $(c|u) \le (d|u)$. As $\bigoplus_u\lambda_u = 0$, there exists $u_0 \in U$ such that $\lambda_{u_0} = 0$. Hence,
        \[
        (c|v \oplus y) = (c| v \oplus u_0)\oplus \left(c\left|\bigoplus_{u\in U\setminus\{u_0\}}\lambda_u\cdot u\right)\right.
        \]
        \[
        \le (d| v \oplus u_0)\oplus \left(d\left|\bigoplus_{u\in U\setminus\{u_0\}}\lambda_u\cdot u\right)\right. = (d|v \oplus y)
        \]
    \item If $(c|v) \le (d|v)$, we claim that for all $u \in U$ and for all $\lambda \le 0$, we have $(c|v \oplus \lambda \cdot u) \le (d | v \oplus \lambda \cdot u)$. The case $(c|u)\le (d|u)$ is immediate. If $(c|u) > (d|u)$, then from $(c|v\oplus u) \le (d|v\oplus u)$, it follows that $(c|u) \le (d|v)$. Hence, for $\lambda \le 0$, we have $(c|\lambda\cdot u) \le (d|v)$, which implies $(c|v \oplus \lambda\cdot u) \le (d|v \oplus \lambda \cdot u)$.
\end{enumerate}

\begin{myexm}\label{exm: v1v2 in HcdU}
    Let $c = (\varepsilon,\varepsilon,0)$, $d = (\varepsilon,1,\varepsilon)$, and $\mathcal{U} = \mathrm{Span}((0,1,1),(0,3,1),(0,1,3))$, given the vectors $v_1 = (0,2,3)$ and $v_2 = (0,1,2)$, we want to decide wether $v_1,v_2 \in \mathscr{H}^{\mathcal{U}}_{c,d}$.

    For the vector $v_1$, we have $v_1 \oplus (0,1,1) = (0,2,3)$, $v_1 \oplus (0,1,3) = (0,2,3)$, $v_2 \oplus (0,3,1) = (0,3,3)$. As we have $(c|(0,2,3)) = 3 = (d|(0,2,3))$ and $(c|(0,3,3)) = 3 < (d|(0,3,3)) = 4$, we have $v_1 \in \mathscr{H}^{\mathcal{U}}_{c,d}$.

    For the vector $v_2$, we have $v_2\oplus (0,1,1) = (0,1,2)$, $v_2 \oplus(0,1,3) = (0,1,3)$ and $v_2 \oplus (0,3,1) = (0,3,2)$. We have $(c|(0,1,2)) = (d|(0,1,2)) = 2$ and $(c|(0,3,2)) = 2 < (d|(0,3,2))=4$, but as $(c|(0,1,3)) = 3 > (d|(0,1,3) )=2$, we conclude that $v_2 \notin \mathscr{H}_{c,d}^\mathcal{U}$.
    We provide a geometric representation in figure \ref{fig: v1v2 in HcdU} for example \ref{exm: v1v2 in HcdU}
\begin{figure}
    \centering
    \includegraphics[width=1\linewidth]{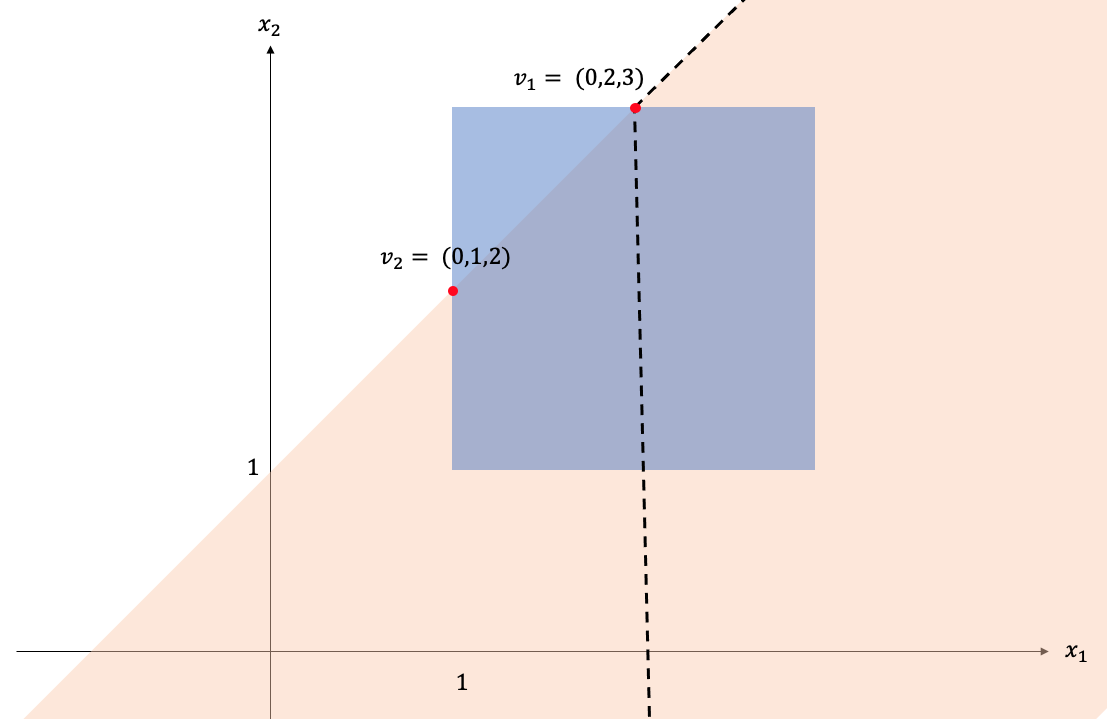}
    \caption{The overview of example \ref{exm: v1v2 in HcdU}. The region in skin color represents the set $\langle c,d\rangle\cap \R^2$, the region in blue represents the set $\mathcal{U}\cap \R^2$, and the region on the right-hand side of the dot line represent the set $\mathscr{H}_{c,d}^\mathcal{U}$. The two red points represent respectively $v_1$ and $v_2$. It is easy to check that $v_2 \notin \mathscr{H}_{c,d}^\mathcal{U}$ and $v_1 \in \mathscr{H}_{c,d}^\mathcal{U}$.} 
    \label{fig: v1v2 in HcdU}
\end{figure}
\end{myexm}

It remains to find the maximal $\rho$ such that $v \oplus  \rho\cdot w \in \mathscr{H}^{\mathcal{U}}_{c,d}$, for $v \in \mathscr{H}^{\mathcal{U}}_{c,d}$ and $w \notin \mathscr{H}^{\mathcal{U}}_{c,d}$. Inspired by the previous method for determining whether a vector $v$ belongs to the tropical cone $ \mathscr{H}_{c,d}^\mathcal{U}$, we remark that the set $\mathscr{H}^\mathcal{U}_{c,d}$ can be rewritten as 
\[\bigcap_{u\in U} \left\{x\in \R^{n+1}\left|(c|x\oplus x_1\cdot u)\le (d|x\oplus x_1\cdot u)
\right\}\right.\]
Let us define the following matrix:
\[
M_u = 
\begin{bmatrix}
    u_1  \oplus 0 &\varepsilon & \varepsilon & \dots & \varepsilon\\
    u_2           & 0          & \varepsilon & \dots & \varepsilon\\
    u_3 & \varepsilon & \ddots & 
    \\
    \vdots            & \vdots & & \ddots\\
    u_{n+1}         &\varepsilon & \dots & \varepsilon & 0
\end{bmatrix}
\]
then one verifies that the inequality $(C_i|x\oplus x_1\cdot u)\le (D_i|x\oplus x_1\cdot u)$ can be rewritten as $C_i^T\otimes M_u \otimes x \le D_i^T \otimes M_u \otimes x$, or simply, $(M_u^T \otimes C_i |x )\le (M_u^T\otimes D_i |x)$, which defines exactly a tropical half-space. Denote $\mathscr{H}_{c,d}^u$ the tropical half-space defined by this equation, then we have \(\mathscr{H}_{c,d}^\mathcal{U} = \bigcap_{u\in U}\mathscr{H}^u_{c,d}\) and we have the following equality $\rho_{v,w}$
\[\max\{\rho| v \oplus \rho\cdot w \in \mathscr{H}_{c,d}^{\mathcal{U}}\} = \min_{u\in U}\left(\max\{\lambda| v \oplus \lambda\cdot w \in \mathscr{H}_{c,d}^u\}\right)\]
It is worth noting that for $u \in U$ such that $w \in \mathscr{H}_{c,d}^u$, we shall have $\max\{\lambda|v\oplus \lambda\cdot w \in \mathscr{H}_{c,d}^u\} = +\infty$. But the right-hand side of the equation is always a finite value since there is at least a $u \in U$ such that $w \notin \mathscr{H}_{c,d}^u$, for such a $u$, we have $(M_u^T\otimes c| w) > (M^T_u\otimes d|w)\ge \varepsilon$, and we can derive that the largest $\lambda$ such that $v \oplus \lambda\cdot w \in \mathscr{H}_{c,d}^u$ is given by 
\(\lambda_u = (M_u^T\otimes d|v)-(M_u^T\otimes c| w) = (d|v\oplus v_1\cdot u) - (c|w \oplus w_1\cdot u)\) and 
$$\rho_{v,w} = \min_{u \in U,\ w \notin \mathscr{H}_{c,d}^u}((d|v\oplus v_1\cdot u) - (c|w \oplus w_1\cdot u))$$
\begin{myexm}
    Consider the case $c = (\varepsilon,\varepsilon,0)$, $d = (\varepsilon,1,\varepsilon)$, $U = \{(0,1,1),(0,1,3),(0,3,1)\}$. $v = (0,3,1)$ and $w = (0,1,2)$. It is easy to verify that $v \in \mathscr{H}_{c,d}^\mathcal{U}$ while $w \notin \mathscr{H}_{c,d}^{\mathcal{U}}$. 

    For $u^{(1)} = (0,1,1),\ u^{(3)} = (0,3,1)$, we have $w \in \mathscr{H}^{u^{(1)}}_{c,d}$ and $w \in \mathscr{H}^{u^{(3)}}_{c,d}$. 

    For $u^{(2)} = (0,1,3)$, we have $(c|w \oplus w_1\cdot u^{(2)}) = 3 > (d|w\oplus w_1\cdot u^{(2)})=2$, thus $w  \notin \mathscr{H}^{u^{(2)}}_{c,d}$, we thus have $\lambda_{u^{(2)}} = (d|v\oplus v_1\cdot u^{(2)}) - (c|w \oplus w_1\cdot u^{(2)})=4-3=1$. We thus have $\rho_{v,w} = 1$. Indeed $v \oplus 1\cdot w = (1,3,3)$ belongs to $\mathscr{H}_{c,d}^\mathcal{U}$ and its on the boundary of $\mathscr{H}_{c,d}^\mathcal{U}$.
\end{myexm}

With these useful tools above, we are able to compute the generating set for $E = \bigcap_{i=1}^N\mathscr{H}_{C_i,D_i}^\mathcal{U}$ with the following steps: 1. We start with the set $V = \{e_i\}_i$, which is a generating set for $\R^n$. 2. For $i$ ranging from $1$ to $N$, update $V$ by the generating set for $\mathrm{Span}(V) \cap \mathscr{H}_{C_i,D_i}$. The resulting set $V$ is then a generating set for $E$.

\begin{remark}[\bf Extremal Points Filtering Procedure]
  In the work of \cite{allamigeon2013computing}, a procedure of extremal points filtering is executed to enhance the performance of the algorithm. For a point $x$ belonging to a tropical cone $\langle C,D \rangle$, one can check wether $x$ is an extremal point by building up the associated directed hypergraph and determine if it has a maximal strongly connected component, for details we invite the reader to refer to \cite{allamigeon2013computing}. Inspired by their approach, we will use the similar method to fasten the computation of extremal points for the tropical cone $E$ defined in section \ref{sec: problem solution}. To reduce the complexity of the algorithm described in the end of section \ref{sec: problem solution} for computation of the generating set $E = \bigcap_{i=1}^N\mathscr{H}_{C_i,D_i}^\mathcal{U}$, we shall insert an extremal points filtering procedure after each iteration: At iteration step $k$, the set $V$ is currently the generating set for the cone $E_k = \bigcap_{i=1}^k\mathscr{H}_i$, we delete points in $V$ that are not extremal in $E_k$ (which still keeps $V$ a generating set), in this way we can reduce the complexity to a large extent. As it has been proved in the previous section that each $\mathscr{H}_{C_i,D_i}^\mathcal{U}$ is a finite intersection of tropical half-spaces $\mathscr{H}_{C_i,D_i}^\mathcal{U}$, $E$ can thus be rewritten as the intersection of finitely many tropical half-spaces and we can apply exactly the same method described in \cite{allamigeon2013computing} to select the extremal points in the generating set of $E$: for $x \in E = \bigcap_j\mathscr{H}_j$, we build the directed hypergraph $\mathcal{G}(x,E)$ associated to $x$ and determine whether $\mathcal{G}(x,E)$ has a maximal strongly connected component (SCC).   
\end{remark}

\section{Case study}
In the following, we use an example to illustrate how we effectuate the backward propagation for the following system $\mathscr{S}$
\[\begin{bmatrix}x^{(k)}_1\\x^{(k)}_2 \end{bmatrix} = \begin{bmatrix}2 & 3\\ 5 & 1\end{bmatrix}\otimes  \begin{bmatrix}x^{(k-1)}_1 \\ x^{(k-1)}_2\end{bmatrix} \oplus \begin{bmatrix} \varepsilon\\ 0\end{bmatrix}\otimes \begin{bmatrix}
    u_k
\end{bmatrix} \oplus \begin{bmatrix}
    w^{(k)}_1 \\ w^{(k)}_2
\end{bmatrix}
\]
where $[w_1^{(k)}, w_2^{(k)}]^T$ ranges in the tropical polyhedra 
$$\mathcal{W} = \left\{\left.\lambda_1\cdot \begin{bmatrix}
    1\\1
\end{bmatrix}\oplus \lambda_2\cdot \begin{bmatrix}
    3\\1
\end{bmatrix} \oplus \lambda_3 \cdot \begin{bmatrix}
    1\\3
\end{bmatrix}\right|\bigoplus_{i=1}^3\lambda_i = 0\right\}$$
and $u$ ranges in the polyhedra $\mathcal{U} = \R$. The safety set is the tropical cone in $\R^2$ defined by $S = \mathrm{Span}([1,0]^T,[0,1]^T)$, or equivalently, in its $\mathcal{M}$-form
\[S = \left\langle 
\begin{bmatrix}
    \varepsilon & \varepsilon & 0 \\
    \varepsilon & -1 &\varepsilon
\end{bmatrix} , 
\begin{bmatrix}
    \varepsilon& 1 & \varepsilon \\
    \varepsilon & \varepsilon & 0
\end{bmatrix}
\right\rangle
\cap \R^2
\]
which we denote as $S = \langle C,D\rangle\cap \R^2$ Given $S$, our goal is to find the backward reach set of $S$ w.r.t the system $\mathscr{S}$. We proceed step by step: we first calculate the set $\phi_\mathcal{W}(S)$, then we calculate $\gamma_\mathcal{U} (\phi_\mathcal{W}(S))$, and finally the backward reach set is given by $A^{-1}(\gamma (\phi_\mathcal{W}(S)))$.

To calculate $\phi_\mathcal{W}(S)$, we follow the steps in described in section \ref{sec: computation phi}. Let $\mathcal{W}'$ be the tropical cone in $\R^3$ generated by $U = \{[0,1,1]^T,\ [0,3,1]^T,\ [0,1,3]^T\}$, then we have $S = S'\cap \R^2$. We first initialize the set $V = \{[0,\varepsilon,\varepsilon]^T,[\varepsilon,0,\varepsilon]^T, [\varepsilon,\varepsilon,0]^T\}$, which is the canonique base for $\R^3$. We then decide the generators for the tropical cone $\mathscr{H}^{C_1,D_1}_\mathcal{W}$, where $C_1 = [\varepsilon,\varepsilon,0],\ D_1 = [\varepsilon,1,\varepsilon]$, which involves calculating the set $V_{in} = V \cap \mathscr{H}^{C_1,D_1}_\mathcal{W} $ and $V_{out} = V \setminus V_{in}$ and combine the vectors of $V_{in}$ and $V_{out}$ to get the generating set for $\mathscr{H}^{C_1,D_1}_\mathcal{W}$. As $[0,\varepsilon,\varepsilon]^T \oplus [0,1,3]^T \notin \langle C_1,D_1\rangle$, we know that $[0,\varepsilon,\varepsilon] \in V_{out}$, and it is easy to verify that $[\varepsilon,0,\varepsilon]^T \in V_{in}$ and $[\varepsilon,\varepsilon,0]^T \in V_{out}$. In conclusion, we have $V_{in} = \{[\varepsilon, 0,\varepsilon]^T\} = \{v_1\}$, $V_{out} = \{[0,\varepsilon,\varepsilon]^T,\ [\varepsilon,\varepsilon,0]^T\} = \{w_1,w_2\}$. The generating set for $\mathscr{H}_\mathcal{W}^{C_1,D_2}$ can be expressed as $\{v_1,\  v_1 \oplus \rho_1\cdot w_1,\ v_1 \oplus \rho_2\cdot w_2\}$, where $\rho_i = \max\{\lambda\mid v_1\oplus \lambda\cdot w_i \in \mathscr{H}_\mathcal{W}^{C_1,D_1}\}$. For $v_1$ and $w_1$, as $v_1^{(1)} = \varepsilon$ and $w_1^{(1)}\neq \varepsilon$, we have $\rho_1 = (D_1|v_1)- (C_1|w_1 \oplus \bigoplus_{u \in U}u) = ([\varepsilon,1,\varepsilon]^T|[\varepsilon,0,\varepsilon]^T) - ([\varepsilon,\varepsilon,0]^T|[0,3,3]^T) = 1-3 = -2$, thus $v_1\oplus \rho_1\cdot w_1 = [-2,0,\varepsilon]^T$. For $v_1$ and $w_2$, as $v_1^{(1)} = \varepsilon$ and $w_2^{(1)} = \varepsilon$, we have $\rho_2 = (D_1|v_1)-(C_1|w_2) = 1-0=1$, thus $v_1\oplus \rho_2 \cdot w_2 = [\varepsilon,0,1]^T$. We find the generating set for $\mathscr{H}_\mathcal{W}^{C_1,D_1}$ is $V_1 = \{[\varepsilon,0,\varepsilon]^T,\ [0,2,\varepsilon]^T,\ [\varepsilon,0,1]^T\}$. Now we aim to calculate the generating set for $ \mathscr{H}_\mathcal{W}^{C_1,D_1}\cap \mathscr{H}_\mathcal{W}^{C_2,D_2}$. Again, we need to decide the set $V_{in}' = V_1 \cap \mathscr{H}_\mathcal{W}^{C_2,D_2}$ and the set $V_{out}' = V_1\setminus V_{in}'$. After calculation, we get $V_{in} = \{[\varepsilon,0,1]^T\}$, which we denote as $V_1 = \{v_1'\}$, and we have $V_{out} = \{[\varepsilon,0,\varepsilon]^T,\ [0,2,\varepsilon]^T\}$, denoted as $V_{out} = \{w_1',w_2'\}$. For $v_1'$ and $w_1'$, we have $\rho_1' = (D_2|v_1') - (C_2|w_1') = ([\varepsilon,\varepsilon,0]^T|[\varepsilon,0,1]^T)-([\varepsilon,-1,\varepsilon]^T|[\varepsilon,0,\varepsilon]^T)=1-(-1)=2$. For $v_1'$ and $w_2'$, we have $\rho_2' = (D_2|v_1') - (C_2|w_2'\oplus \bigoplus_{u\in U}u) = ([\varepsilon,\varepsilon,0]^T|[\varepsilon,0,1]^T) - ([\varepsilon,-1,\varepsilon]^T|[0,2,\varepsilon]^T\oplus [1,3,3]^T) = -1$. A generating set for $ \mathscr{H}_\mathcal{W}^{C_1,D_1}\cap \mathscr{H}_\mathcal{W}^{C_2,D_2}$ is $\{v_1',v_1'\oplus \rho_1'\cdot w_1', v_1'\oplus \rho_2'\cdot w_2'\} = \{[\varepsilon,0,1]^T,[\varepsilon,2,1]^T,[-1,1,1]^T\}$, which we normalize to be $\{[\varepsilon,0,1]^T,[\varepsilon,2,1]^T,[0,2,2]^T\}$. A geometric representation for the example is represented in figure \ref{fig: final figure}. 

By following the procedure described in section \ref{sec: problem solution}, we get finally the backward reachable set
\[\begin{aligned}
    \Upsilon(S) &= A^{-1}\circ \gamma_\mathcal{U} \circ \phi_\mathcal{W}(S)\\
    &=A^{-1}\circ \gamma_\mathcal{U}\left(\mathrm{Span}\left(    \begin{bmatrix}
        \varepsilon \\ 0 \\ 1
    \end{bmatrix},
    \begin{bmatrix}
        \varepsilon \\ 2 \\ 1
    \end{bmatrix},
    \begin{bmatrix}
        0 \\ 2 \\ 2
    \end{bmatrix}
    \right)\cap \R^2\right)\\
    & = A^{-1}\left(\mathrm{Span}\left(\begin{bmatrix}
        \varepsilon \\ 0 \\ \varepsilon
    \end{bmatrix},
    \begin{bmatrix}
        \varepsilon \\ 0 \\ 1
    \end{bmatrix}
    \begin{bmatrix}
        0 \\ 2 \\ 2
    \end{bmatrix}
    \right)\right)\\
    & = \mathrm{Conv}\left(
    \begin{bmatrix}
        0\\ \varepsilon
    \end{bmatrix}
    ,
    \begin{bmatrix}
        -2\\ 1
    \end{bmatrix}
    ,
    \begin{bmatrix}
        1\\ 3
    \end{bmatrix}
    ,
    \begin{bmatrix}
        -2\\ 3
    \end{bmatrix}
    \right)
\end{aligned}
 \]

\begin{figure}
  \centering
  \includegraphics[width=1\linewidth]{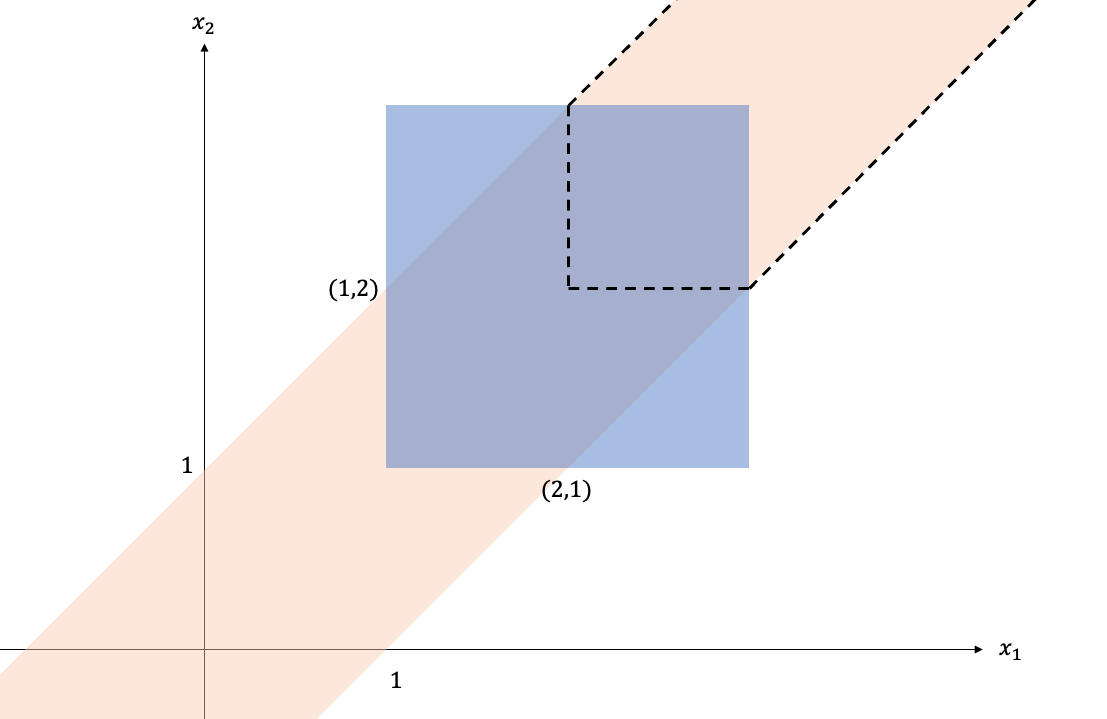}
  \caption{The overview of the three regions: $S$, $\mathcal{W}$ and $\phi_{\mathcal{W}}(S)$. The skin color region represents $S$, $\mathcal{W}$ corresponds to the blue region, and $\phi_{\mathcal{W}}(S)$ is surrounded by the line of dots.}
  \label{fig: final figure}
\end{figure}

\section{Conclusion}
In this paper, we investigated the problem of backward reachability analysis for uncertain max-plus linear systems. Unlike stochastic approaches that describe uncertainty through probability distributions, we considered a deterministic framework in which uncertainty sets are represented by tropical polyhedra. This representation provides a flexible and expressive way to characterize bounded uncertainties within the max-plus algebraic setting. Future work may focus on improving the computational efficiency of the proposed method, possibly at the cost of some loss of precision, and on addressing uncertainties in the system parameters.

\bibliographystyle{ieeetr}
\bibliography{cite.bib}
\addtolength{\textheight}{-12cm}   




\end{document}